\pdfoutput=1
\PassOptionsToPackage{unicode=true}{hyperref} 

\documentclass[acmsmall,screen]{acmart}
\usepackage[T1]{fontenc}
\usepackage[utf8]{inputenc}

\usepackage{hyperref}
\hypersetup{
            pdftitle={Denotational Semantics and a Fast Interpreter for jq},
            pdfborder={0 0 0},
            breaklinks=true}
\usepackage{color}
\usepackage{fancyvrb}

\DefineVerbatimEnvironment{Highlighting}{Verbatim}{commandchars=\\\{\}}
\newenvironment{Shaded}{}{}

\newcommand{\CommentTok}[1]{\textcolor[rgb]{0.38,0.63,0.69}{\textit{#1}}}

\newcommand{\DataTypeTok}[1]{\textcolor[rgb]{0.56,0.13,0.00}{#1}}
\newcommand{\DecValTok}[1]{\textcolor[rgb]{0.25,0.63,0.44}{#1}}

\newcommand{\KeywordTok}[1]{\textcolor[rgb]{0.00,0.44,0.13}{\textbf{#1}}}
\newcommand{\NormalTok}[1]{#1}
\newcommand{\OperatorTok}[1]{\textcolor[rgb]{0.40,0.40,0.40}{#1}}

\newcommand{\PreprocessorTok}[1]{\textcolor[rgb]{0.74,0.48,0.00}{#1}}

\providecommand{\tightlist}{}

\acmJournal{PACMPL}
\acmVolume{1}
\acmNumber{CONF} 
\acmArticle{1}
\acmYear{2018}
\acmMonth{1}
\acmDOI{} 
\startPage{1}

\setcopyright{none}

\usepackage{tikz}
\usetikzlibrary{calc,fit,patterns.meta,positioning}

\makeatletter
\DeclareRobustCommand{\rvdots}{%
  \vbox{
    \baselineskip4\p@\lineskiplimit\z@
    \kern-\p@
    \hbox{.}\hbox{.}\hbox{.}
  }}
\makeatother

\tikzset{filter/.style={rounded corners,draw,inner xsep=8}}
\tikzset{clone/.style={double}}

\usepackage{pgfplots}
\pgfplotsset{compat=1.9}
\usepgfplotslibrary{units}

\usepackage{subcaption}

\usepackage{mathtools}

\usepackage{newfloat}
\DeclareFloatingEnvironment{listing}

\usepackage{cleveref}

\usepackage{operators}
\usepackage{includetexgraphics}

\usepackage[]{natbib}
\title{Denotational Semantics and a Fast Interpreter for jq}
\author{Michael Färber}
\orcid{0000-0003-1634-9525}
\affiliation{
  \institution{University of Innsbruck}
  \city{Innsbruck}
  \country{Austria}
}
\email{michael.faerber@uibk.ac.at}
\date{}

\begin{document}
\begin{abstract}
jq is a widely used tool that provides a programming language to
manipulate JSON data. However, its semantics are currently only
specified by its implementation, making it difficult to reason about its
behaviour. To this end, I provide a syntax and denotational semantics
for a subset of the jq language. In particular, the semantics provide a
new way to interpret updates. I implement an extended version of the
semantics in a novel interpreter for the jq language called jaq.
Although jaq uses a significantly simpler approach to execute jq
programs than jq, jaq is faster than jq on ten out of thirteen
benchmarks.
\end{abstract}
\maketitle

\hypertarget{introduction}{%
\section{Introduction}\label{introduction}}

UNIX has popularised the concept of \emph{filters} and \emph{pipes}
\citep{DBLP:journals/bstj/Ritchie84}: A filter is a program that reads
from an input stream and writes to an output stream. Pipes are used to
compose filters.

JSON (JavaScript Object Notation) is a widely used data serialisation
format \citep{rfc8259}. A JSON value is either null, a boolean, a
number, a string, an array of values, or an associative map from strings
to values.

jq is a tool that provides a language to define filters and an
interpreter to execute them. Where UNIX filters operate on streams of
characters, jq filters operate on streams of JSON values. This allows to
manipulate JSON data with relatively compact filters. For example, given
as input the public JSON dataset of streets in Paris
\citep{paris-voies}, jq retrieves the number of streets (6508) with the
filter ``\texttt{length}'', the names of the streets with the filter
``\texttt{.{[}{]}.fields.nomvoie}'', and the total length of all streets
(1576813 m) with the filter
``\texttt{{[}.{[}{]}.fields.longueur{]}\ \textbar{}\ add}''. jq provides
a Turing-complete language that is interesting on its own; for example,
``\texttt{{[}0,\ 1{]}\ \textbar{}\ recurse({[}.{[}1{]},\ add{]}){[}0{]}"}
generates the stream of Fibonacci numbers. This makes jq a widely used
tool. I refer to the program jq as''jq" and to its language as ``the jq
language''.

The semantics of the jq language are only informally specified in the jq
manual \citep{jq-manual}. However, the documentation frequently does not
cover certain cases, or the implementation downright contradicts the
documentation. For example, the documentation states that the filter
\texttt{limit(n;\ f)} ``extracts up to \texttt{n} outputs from
\texttt{f}''. However, \texttt{limit(0;\ f)} extracts up to 1 outputs
from \texttt{f}, and for negative values of \texttt{n},
\texttt{limit(n;\ f)} extracts all outputs of \texttt{f}. While this
particular example could easily be corrected, the underlying issue of
having no formally specified semantics to rely on remains. Having such
semantics also allows to determine whether certain behaviour of the
implementation is accidental or intended.

However, a formal specification of the behaviour of jq would be very
verbose, because jq has many special cases whose merit is not apparent.
Therefore, I have striven to create denotational semantics
(\autoref{specification}) that closely resemble those of jq such that in
most cases, their behaviour coincides, whereas they may differ in more
exotic cases. One particular improvement over jq are the new update
semantics (\autoref{updates}), which are simpler to describe and
implement, eliminate a range a potential errors, and allow for more
performant execution.

Before executing a filter, jq compiles it to a relatively low-level
representation. Based on the formal semantics, I wanted to find out how
fast a simple interpreter could be that operates on an only slightly
processed abstract syntax tree of the filter. To this end, I implemented
an interpreter called \emph{jaq} (\autoref{implementation}). One
particular challenge for performant execution turned out to be how to
avoid clones of data (\autoref{cloning}) and how to restrict the scope
of clones (\autoref{strictness}) in order to allow for mutation.

To evaluate jaq and jq, I wrote a number of benchmarks
(\autoref{evaluation}). Furthermore, I evaluated jaq and jq using an
existing interpreter written in the jq language for the Turing-complete
language Brainfuck. The surprising result is that even with its naive
execution approach, jaq is faster than jq on ten out of thirteen
benchmarks.

\hypertarget{preliminaries}{%
\section{Preliminaries}\label{preliminaries}}

This goal of this section is to convey an intuition about how jq
functions. In \autoref{specification}, I will formally specify syntax
and semantics that closely approximate those of jq. The official
documentation of jq is \citep{jq-manual}.

jq programs are called \emph{filters}. For now, let us consider a filter
to be a function from a value to a (lazy) stream of values. Furthermore,
let us assume a value to be either a boolean, an integer, or an array of
values.

The identity filter ``\texttt{.}'' returns a stream containing the
input.

Arithmetic operations, such as addition, subtraction, multiplication,
division, and remainder, are available in jq. For example,
``\texttt{.\ +\ 1}'' returns a stream containing the successor of the
input. Here, ``\texttt{1}'' is a filter that returns the value
\texttt{1} for any input.

Concatenation is an important operator in jq: The filter
``\texttt{f,\ g}'' concatenates the outputs of the filters \texttt{f}
and \texttt{g}. For example, the filter ``\texttt{.,\ .}'' returns a
stream containing the input value twice.

Composition is one of the most important operators in jq: The filter
``\texttt{f\ \textbar{}\ g}'' maps the filter \texttt{g} over all
outputs of the filter \texttt{f}. For example,
``\texttt{(1,\ 2,\ 3)\ \textbar{}\ (.\ +\ 1)}'' returns
\texttt{2,\ 3,\ 4}.

Arrays are created from a stream produced by \texttt{f} using the filter
``\texttt{{[}f{]}}''. For example, the filter
``\texttt{{[}1,\ 2,\ 3{]}}'' concatenates the output of the filters
``\texttt{1}'', ``\texttt{2}'', and ``\texttt{3}'' and puts it into an
array, yielding the value \texttt{{[}1,\ 2,\ 3{]}}. The inverse filter
``\texttt{.{[}{]}}'' returns a stream containing the values of an array
if the input is an array. For example, running ``\texttt{.{[}{]}}'' on
the array \texttt{{[}1,\ 2,\ 3{]}} yields the stream \texttt{1,\ 2,\ 3}
consisting of three values. We can combine the two shown filters to map
over arrays; for example, when given the input \texttt{{[}1,\ 2,\ 3{]}},
the filter ``\texttt{{[}.{[}{]}\ \textbar{}\ (.\ +\ 1){]}}'' returns a
single value \texttt{{[}2,\ 3,\ 4{]}}. The values of an array at indices
produced by \texttt{f} are returned by ``\texttt{.{[}f{]}}''. For
example, given the input \texttt{{[}1,\ 2,\ 3{]}}, the filter
``\texttt{.{[}0,\ 2,\ 0{]}}'' returns the stream \texttt{1,\ 3,\ 1}.

Case distinctions can be performed with the filter
``\texttt{if\ f\ then\ g\ else\ h\ end}''. For every value \texttt{v}
produced by \texttt{f}, this filter returns the output of \texttt{g} if
\texttt{v} is true and the output of \texttt{h} otherwise. For example,
given the input \texttt{1}, the filter
``\texttt{if\ (.\ \textless{}\ 1,\ .\ ==\ 1,\ .\ \textgreater{}=\ 1)\ then\ .\ else\ {[}{]}\ end}''
returns \texttt{{[}{]},\ 1,\ 1}.

Fix points are calculated as follows: Given a filter \texttt{f},
``\texttt{recurse(f)}'' returns the output of
``\texttt{.,\ (f\ \textbar{}\ recurse(f))}''. This way, we can define a
filter to calculate the factorial function, for example.

\begin{example}[Factorial]\label{ex:fac}Let us define a filter
\texttt{fac} that should return \(n!\) for any input number \(n\). We
will define \texttt{fac} using the fix point of a filter
\texttt{update}. The input and output of \texttt{update} shall be an
array \texttt{{[}n,\ acc{]}}, satisfying the invariant that the final
output is \texttt{acc} times the factorial of \texttt{n}. The initial
value passed to \texttt{update} is the array ``\texttt{{[}.,\ 1{]}}''.
We can retrieve \texttt{n} from the array with ``\texttt{.{[}0{]}}'' and
\texttt{acc} with ``\texttt{.{[}1{]}}''. We can now define
\texttt{update} as
``\texttt{if\ .{[}0{]}\ \textgreater{}\ 1\ then\ {[}.{[}0{]}\ -\ 1,\ .{[}0{]}\ *\ .{[}1{]}{]}\ else\ empty\ end}'',
where ``\texttt{empty}'' is a filter that returns an empty stream. Given
the input value \texttt{4}, the filter
``\texttt{{[}.,\ 1{]}\ \textbar{}\ recurse(update)}'' returns
\texttt{{[}4,\ 1{]},\ {[}3,\ 4{]},\ {[}2,\ 12{]},\ {[}1,\ 24{]}}. We
are, however, only interested in the accumulator contained in the last
value. So we can write
``\texttt{{[}.,\ 1{]}\ \textbar{}\ last(recurse(update))\ \textbar{}\ .{[}1{]}}'',
where ``\texttt{last(f)}'' is a filter that outputs the last output of
\texttt{f}. This then yields a single value \texttt{24} as
result.\end{example}

Composition can also be used to bind values to \emph{variables}. The
filter ``\texttt{f\ as\ \$x\ \textbar{}\ g}'' performs the following:
Given an input value \texttt{i}, for every output \texttt{o} of the
filter \texttt{f} applied to \texttt{i}, the filter binds the variable
\texttt{\$x} to the value \texttt{o}, making it accessible to
\texttt{g}, and yields the output of \texttt{g} applied to the original
input value \texttt{i}. For example, the filter
``\texttt{(0,\ 2)\ as\ \$x\ \textbar{}\ ((1,\ 2)\ as\ \$y\ \textbar{}\ (\$x\ +\ \$y))}''
yields the stream \texttt{1,\ 2,\ 3,\ 4}. Note that in this particular
case, we could also write this as ``\texttt{(0,\ 2)\ +\ (1,\ 2)}'',
because arithmetic operators such as ``\texttt{f\ +\ g}'' take as inputs
the Cartesian product of the output of \texttt{f} and
\texttt{g}.\footnote{Haskell users might appreciate the similarity of
  the two filters to their Haskell analoga
  ``\texttt{{[}0,\ 2{]}\ \textgreater{}\textgreater{}=\ (\textbackslash{}x\ -\textgreater{}\ {[}1,\ 2{]}\ \textgreater{}\textgreater{}=\ (\textbackslash{}y\ -\textgreater{}\ return\ (x+y)))}''
  and
  ``\texttt{(+)\ \textless{}\$\textgreater{}\ {[}0,\ 2{]}\ \textless{}*\textgreater{}\ {[}1,\ 2{]}}'',
  which both return \texttt{{[}1,\ 2,\ 3,\ 4{]}}.} However, there are
cases where variables are indispensable.

\begin{example}[Variables Are Necessary]jq defines a filter
``\texttt{inside(xs)}'' that expands to
``\texttt{.\ as\ \$x\ \textbar{}\ xs\ \textbar{}\ contains(\$x)}''.
Here, we wish to pass \texttt{xs} as input to \texttt{contains}, but at
the same point, we also want to pass the input given to \texttt{inside}
as an argument to \texttt{contains}. Without variables, we could not do
both.\end{example}

Folding over streams can be done using \texttt{reduce} and
\texttt{foreach}: The filter
``\texttt{reduce\ xs\ as\ \$x\ (init;\ f)}'' keeps a state that is
initialised with the output of \texttt{init}. For every element
\texttt{\$x} yielded by the filter \texttt{xs}, \texttt{reduce} feeds
the current state to the filter \texttt{f}, which may reference
\texttt{\$x}, then sets the state to the output of \texttt{f}. When all
elements of \texttt{xs} have been yielded, \texttt{reduce} returns the
current state. For example, the filter
``\texttt{reduce\ .{[}{]}\ as\ \$x\ (0;\ .\ +\ \$x)}'' calculates the
sum over all elements of an array. Similarly,
``\texttt{reduce\ .{[}{]}\ as\ \$x\ (0;\ .\ +\ 1)}'' calculates the
length of an array. These two filters are called ``\texttt{add}'' and
``\texttt{length}'' in jq, and they allow to calculate the average of an
array by ``\texttt{add\ /\ length}''. The filter
``\texttt{foreach\ xs\ as\ \$x\ (init;\ f)}'' is similar to
\texttt{reduce}, but also yields all intermediate states, not only the
last state. For example,
``\texttt{foreach\ .{[}{]}\ as\ \$x\ (0;\ .\ +\ \$x)}'' yields the
cumulative sum over all array elements.

Updating values can be done with the operator ``\texttt{\textbar{}=}'',
which has a similar function as lens setters in languages such as
Haskell
\citep{DBLP:conf/popl/FosterGMPS05, DBLP:journals/programming/PickeringGW17}:
Intuitively, the filter ``\texttt{p\ \textbar{}=\ f}'' considers any
value \texttt{v} returned by \texttt{p} and replaces it by the output of
\texttt{f} applied to \texttt{v}.{[}\^{}haskell-lenses{]} We call a
filter on the left-hand side of ``\texttt{\textbar{}=}'' a \emph{path
expression}. For example, when given the input \texttt{{[}1,\ 2,\ 3{]}},
the filter ``\texttt{.{[}{]}\ \textbar{}=\ (.\ +\ 1)}'' yields
\texttt{{[}2,\ 3,\ 4{]}}, and the filter
``\texttt{.{[}1{]}\ \textbar{}=\ (.\ +\ 1)}'' yields
\texttt{{[}1,\ 3,\ 3{]}}. We can also nest these filters; for example,
when given the input \texttt{{[}{[}1,\ 2{]},\ {[}3,\ 4{]}{]}}, the
filter
``\texttt{(.{[}{]}\ \textbar{}\ .{[}{]})\ \textbar{}=\ (.\ +\ 1)}''
yields \texttt{{[}{[}2,\ 3{]},\ {[}4,\ 5{]}{]}}. However, not every
filter is a path expression; for example, the filter ``\texttt{1}'' is
not a path expression because ``\texttt{1}'' does not point to any part
of the input value but creates a new value.

Identities such as ``\texttt{.{[}{]}\ \textbar{}=\ f}'' being equivalent
to ``\texttt{{[}.{[}{]}\ \textbar{}\ f{]}}'', or
``\texttt{.\ \textbar{}=\ f}'' being equivalent to \texttt{f}, would
allow defining the behaviour of updates. However, these identities do
not hold in jq due the way it handles filters \texttt{f} that return
multiple values. In particular, when we pass \texttt{0} to the filter
``\texttt{.\ \textbar{}=\ (1,\ 2)}'', the output is \texttt{1}, not
\texttt{(1,\ 2)} as we might have expected. Similarly, when we pass
\texttt{{[}1,\ 2{]}} to the filter
``\texttt{.{[}{]}\ \textbar{}=\ (.,\ .)}'', the output is
\texttt{{[}1,\ 2{]}}, not \texttt{{[}1,\ 1,\ 2,\ 2{]}} as expected. This
behaviour of jq is cumbersome to define and to reason about. This
motivates in part the definition of more simple and elegant semantics
that behave like jq in most typical use cases but eliminate corner cases
like the ones shown.

\hypertarget{specification}{%
\section{Specification}\label{specification}}

This section describes syntax and semantics for a subset of the jq
language. To set the formal syntax and semantics apart from the concrete
syntax introduced in \autoref{preliminaries}, we use cursive font (as in
``\(f\)'', ``\(v\)'') for the specification instead of the previously
used typewriter font (as in ``\texttt{f}'', ``\texttt{v}'').

\hypertarget{syntax}{%
\subsection{Syntax}\label{syntax}}

A \emph{filter} \(f\) is defined by
\[f \coloneqq n \mid \$x \mid . \mid .[] \mid .[f] \mid [f] \mid (f) \mid f? \mid f \star f \mid f \circ f \mid \ifj f \thenj f \elsej f \mid x \mid x(f; \dots; f)\]
where \(n\) is an integer and \(x\) is an identifier (such as
``empty''). We call \(\$x\) a variable. By convention, we write \(\$x'\)
to denote a fresh variable. The potential instances of \(\star\) and
\(\circ\) are given in \autoref{tab:binops}. Furthermore, \(f\) can be a
variable binding of the shape ``\(f \as \$x \mid f\)'' or a fold of the
shape ``\(\phi\; f \as \$x (f; f)\)'', where \(\phi\) is either
``reduce'' or ``foreach''.

\begin{table}

\caption{Binary operators, given in order of increasing precedence.
Operators surrounded by parentheses have equal
precedence.\label{tab:binops}}

\begin{tabular}{lll}

\toprule

Name & Symbol & Operators\tabularnewline

\midrule

Complex & \(\star\) & ``\(\mid\)'', ``,'', ``\(\models\)'', ``or'',
``and''\tabularnewline
Cartesian & \(\circ\) & (\(=\), \(\neq\)), (\(<\), \(\leq\), \(>\),
\(\geq\)), (\(+\), \(-\)), (\(*\), \(/\)), \(\%\)\tabularnewline

\bottomrule

\end{tabular}

\end{table}

A \emph{filter definition} has the shape
``\(f(x_1; \dots; x_n) \coloneqq g\)''. Here, \(f\) is an \(n\)-ary
filter where \(g\) may refer to \(x_i\). For example, this allows us to
define filters that produce the booleans, by defining
\(\true \coloneqq (0 = 0)\) and \(\false \coloneqq (0 \neq 0)\).

A value \(v\) is defined by
\[v \coloneqq \true \mid \false \mid n \mid [v, \dots, v]\] where \(n\)
is an integer. While this captures only a subset of JSON values, it
provides a solid base to specify semantics such that they are relatively
straightforward to extend to the full set of JSON values.

\hypertarget{semantics}{%
\subsection{Semantics}\label{semantics}}

The goals for creating these semantics were, in descending order of
importance:

\begin{itemize}
\tightlist
\item
  Simplicity: The semantics should be easy to describe and implement.
\item
  Performance: The semantics should allow for performant execution.
\item
  Compatibility: The semantics should be consistent with jq.
\end{itemize}

Let us start with a few definitions. A context is a mapping from
variables to values. A value result is either a value or an error
\(\bot\). A stream of value results is written as
\(\langle v_0, \dots, v_n \rangle\). The concatenation of two streams
\(s_1\), \(s_2\) is written as \(s_1 + s_2\).

The \emph{evaluation} of a filter \(f\) with a context \(c\) and a value
result \(v\) is denoted as \(f|^c_v\) and returns a stream of value
results. We impose for any filter \(f\),
\(f|^c_\bot = \langle \bot \rangle\), and define its evaluation
semantics only on values. We say that two filters \(f, g\) are
equivalent iff \(f|^c_v = g|^c_v\) for all \(c\) and \(v\).

We are now going to introduce a few helper functions. The first function
transform a stream into an array if all stream elements are values, or
into the leftmost error\footnote{In these simplified semantics, we have
  only a single kind of error, \(\bot\), so it might seem pointless to
  specify which error we return. However, in an implementation, we may
  have different kinds of errors.} in the stream otherwise:
\[[\langle v_0, \dots, v_n\rangle] = \begin{cases}
\left[v_0, \dots, v_n\right] & \text{if for all  } i, v_i \neq \bot \\
v_{\min\{i \mid v_i = \bot\}} & \text{otherwise} \\
\end{cases}\] The next function helps define filters such as
if-then-else, conjunction, and disjunction:
\[\ite(v, i, t, e) = \begin{cases}
\langle \bot \rangle & \text{if } v = \bot \\
t & \text{if } v \neq \bot \text{ and } v = i \\
e & \text{otherwise}
\end{cases}\] The last function serves to retrieve the \(i\)-th element
from a list, if it exists: \[v[i] = \begin{cases}
v_i & \text{if } v = [v_0, \dots, v_n] \text{ and } 0 \leq i < n \\
\bot & \text{otherwise}
\end{cases}\]

To evaluate calls to filters that have been introduced by definition, we
define the substitution \(\varphi[f_1 / x_1, \dots, f_n / x_n]\) to be
\(\sigma \varphi\), where
\(\sigma = \left\{x_1 \mapsto f_1, \dots, x_n \mapsto f_n\right\}\). The
substitution \(\sigma \varphi\) is defined in
\autoref{tab:substitution}: It both applies the substitution \(\sigma\)
and replaces all variables bound in \(\varphi\) by fresh ones. This
prevents variable bindings in \(\varphi\) from shadowing variables that
occur in the co-domain of \(\sigma\).

\begin{example}Consider the filter ``\(0 \as \$x \mid f(\$x)\)'', where
``\(f(g) \coloneqq 1 \as \$x \mid g\)''. Here, ``\(f(\$x)\)'' expands to
``\(1 \as \$x' \mid \$x\)'', where ``\(\$x'\)'' is a fresh variable. The
whole filter expands to ``\(0 \as \$x \mid 1 \as \$x' \mid \$x\)'',
which evaluates to 0. If we would (erroneously) fail to replace \(\$x\)
in \(f(g)\) by a fresh variable, then the whole filter would expand to
``\(0 \as \$x \mid 1 \as \$x \mid \$x\)'', which evaluates to
1.\end{example}

\begin{table}

\caption{Substitution. Here, \(\$x'\) is a fresh variable and
\(\sigma' = \sigma\left\{\$x \mapsto \$x'\right\}\).
\label{tab:substitution}}

\begin{tabular}{ll}

\toprule

\(\varphi\) & \(\sigma\varphi\)\tabularnewline

\midrule

\(.\), \(n\) (where \(n \in \mathbb{Z}\)), or
\(.[]\) & \(\varphi\)\tabularnewline
\(\$x\) or \(x\) & \(\sigma(\varphi)\)\tabularnewline
\(.[f]\) & \(.[\sigma f]\)\tabularnewline
\(f?\) & \((\sigma f)?\)\tabularnewline
\(f \star g\) & \(\sigma f \star \sigma g\)\tabularnewline
\(f \circ g\) & \(\sigma f \circ \sigma g\)\tabularnewline
\(\ifj f \thenj g \elsej h\) & \(\ifj \sigma f \thenj \sigma g \elsej \sigma h\)\tabularnewline
\(x(f_1; \dots; f_n)\) & \(x(\sigma f_1; \dots; \sigma f_n)\)\tabularnewline
\(f \as \$x \mid g\) & \(\sigma f \as \$x' \mid \sigma' g\)\tabularnewline
\(\phi\; xs \as \$x (init; f)\) & \(\phi\; \sigma xs \as \$x'(\sigma init; \sigma' f)\)\tabularnewline

\bottomrule

\end{tabular}

\end{table}

\begin{table}

\caption{Evaluation semantics. \label{tab:eval-semantics}}

\begin{tabular}{ll}

\toprule

\(\varphi\) & \(\varphi|^c_v\)\tabularnewline

\midrule

\(\emptys\) & \(\langle \rangle\)\tabularnewline
\(.\) & \(\langle v \rangle\)\tabularnewline
\(n\) (where
\(n \in \mathbb{Z}\)) & \(\langle n \rangle\)\tabularnewline
\(\$x\) & \(\langle c(\$x) \rangle\)\tabularnewline
\([f]\) & \(\langle \left[f|^c_v\right] \rangle\)\tabularnewline
\(f, g\) & \(f|^c_v + g|^c_v\)\tabularnewline
\(f \mid g\) & \(\sum_{x \in f|^c_v} g|^c_x\)\tabularnewline
\(f \as \$x \mid g\) & \(\sum_{x \in f|^c_v} g|^{c\{\$x \mapsto x\}}_v\)\tabularnewline
\(f \circ g\) & \(\sum_{x \in f|^c_v} \sum_{y \in g|^c_v} \langle x \circ y \rangle\)\tabularnewline
\(f?\) & \(\sum_{x \in f|^c_v} \begin{cases}\langle \rangle & \text{if } x = \bot \\ \langle x \rangle & \text{otherwise}\end{cases}\)\tabularnewline
\(f \andj g\) & \(\sum_{x \in f|^c_v} \ite(x, \false, \langle \false \rangle, g|^c_v)\)\tabularnewline
\(f \orj g\) & \(\sum_{x \in f|^c_v} \ite(x, \true, \langle \true \rangle, g|^c_v)\)\tabularnewline
\(\ifj f \thenj g \elsej h\) & \(\sum_{x \in f|^c_v} \ite(x, \true, g|^c_v, h|^c_v)\)\tabularnewline
\(.[]\) & \(\begin{cases}\langle v_0, \dots, v_n\rangle & \text{if } v = [v_0, \dots, v_n] \\ \langle \bot \rangle & \text{otherwise}\end{cases}\)\tabularnewline
\(.[f]\) & \(\sum_{i \in f|^c_v} \langle v[i] \rangle\)\tabularnewline
\(\phi\; xs \as \$x (init; f)\) & \(\sum_{i \in init|^c_v} \phi^c_i(xs|^c_v, f)\)\tabularnewline
\(x(f_1; \dots; f_n)\) & \(g[f_1 / x_1, \dots, f_n / x_n]|^c_v\) if
\(x(x_1; \dots; x_n) \coloneqq g\)\tabularnewline
\(f \models g\) & see \autoref{tab:update-semantics}\tabularnewline

\bottomrule

\end{tabular}

\end{table}

The evaluation semantics are given in \autoref{tab:eval-semantics}. We
suppose that the Cartesian operator \(\circ\) is defined on pairs of
values, yielding a value result. We have seen examples of the shown
filters in \autoref{preliminaries}. The semantics diverge relatively
little from the implementation in jq. One notable exception is
\(f \circ g\), which jq evaluates differently as
\(\sum_{y \in g|^c_v} \sum_{x \in f|^c_v} \langle x \circ y \rangle\).
The reason will be given in \autoref{cloning}. Note that the difference
only shows when both \(f\) and \(g\) return multiple values.

\[\phi^c_v(xs, f) \coloneqq \begin{cases}
\langle \phantom{v} \rangle + \sum_{x \in f|^{c\{\$x \mapsto x\}}_v}\phi^c_x(xt, f) & \text{if } xs = \langle x \rangle + xt \text{ and } \phi = \reduce \\
\langle v \rangle + \sum_{x \in f|^{c\{\$x \mapsto x\}}_v}\phi^c_x(xt, f) & \text{if } xs = \langle x \rangle + xt \text{ and } \phi = \foreachj \\
\langle v \rangle & \text{otherwise}
\end{cases}\]

In addition to the filters defined in \autoref{tab:eval-semantics}, we
define the semantics of the two fold-like filters ``reduce'' and
``foreach'' as follows, where \(xs\) evaluates to
\(\langle x_0, \dots, x_n \rangle\): \begin{align*}
\reduce   xs \as \$x\; (init;\, f) &= init &
\foreachj xs \as \$x\; (init;\, f) &= init \\
& \mid x_0 \as \$x \mid f &
& \mid ., (x_0 \as \$x \mid f \\
& \mid \dots &
& \mid \dots \\
& \mid x_n \as \$x \mid f &
& \mid ., (x_n \as \$x \mid f)\dots) \\
\end{align*} Both filters fold \(f\) over the sequence \(xs\) with the
initial value \(init\). Their main difference is that ``reduce'' returns
only the final value(s), whereas ``foreach'' also returns all
intermediate ones.

The following property can be used to eliminate bindings.

\begin{lemma}Let \(\varphi(f)\) be a filter such that
\(\varphi(f)|^c_v\) has the shape ``\(\sum_{x \in f|^c_v} \dots\)''.
Then \(\varphi(f)\) is equivalent to
``\(f \as \$x \mid \varphi(\$x)\)''.\end{lemma}

\begin{proof}We have to prove the statement for \(\varphi(f)\) set to
``\(f \mid g\)'', ``\(f \as \$x \mid g\)'', ``\(f \circ g\)'',
``\(f?\)'', ``\(f \andj g\)'', ``\(f \orj g\)'',
``\(\ifj f \thenj g \elsej h\)'', ``\(.[f]\)'', and
``\(\phi\; xs \as \$x(f; g)\)''. Let us consider the filter
\(\varphi(f)\) to be \(.[f]\). Then we show that \(.[f]\) is equivalent
to \(f \as \$x \mid .[\$x]\): \begin{align*}
  (f \as \$x \mid .[\$x])|^c_v
  &= \sum_{x \in f|^c_v} .[\$x]|^{c\{\$x \mapsto x\}}_v \\
  &= \sum_{x \in f|^c_v} \sum_{i \in \$x|^{c\{\$x \mapsto x\}}_v} \langle v[i] \rangle \\
  &= \sum_{x \in f|^c_v} \sum_{i \in \langle x \rangle} \langle v[i] \rangle \\
  &= \sum_{x \in f|^c_v} \langle v[x] \rangle \\
  &= .[f]|^c_v
  \end{align*} The other cases for \(\varphi(f)\) can be proved
similarly.\end{proof}

The semantics of jq and those shown in \autoref{tab:eval-semantics}
differ most notably in the case of updates; that is, \(f \models g\). We
are going to deal with this in the next subsection.

\hypertarget{updates}{%
\subsection{Updates}\label{updates}}

jq's update mechanism works with \emph{paths}. A path is a sequence of
indices \(i_j\) that can be written as \(.[i_1]\dots[i_n]\). It refers
to a value that can be retrieved by the filter
``\(.[i_1] \mid \dots \mid .[i_n]\)''. Note that ``\(.\)'' is a valid
path, referring to the input value.

The update operation ``\(f \models g\)'' attempts to first obtain the
paths of all values returned by \(f\), then for each path, it replaces
the value at the path by \(g\) applied to it. Note that \(f\) is not
allowed to produce new values; it may only return paths.

\begin{example}Consider the input value \([[1, 2], [3, 4]]\). We can
retrieve the arrays \([1, 2]\) and \([3, 4]\) from the input with the
filter ``\(.[]\)'', and we can retrieve the numbers 1, 2, 3, 4 from the
input with the filter ``\(.[] \mid .[]\)''. To replace each number with
its successor, we run ``\((.[] \mid .[]) \models .+1\)'', obtaining
\([[2, 3], [4, 5]]\). Internally, in jq, this first builds the paths
\(.[0][0]\), \(.[0][1]\), \(.[1][0]\), \(.[1][1]\), then updates the
value at each of these paths with \(g\).\end{example}

There are several problems with this approach to updates: One of these
problems is that if \(g\) returns no output, the collected paths may
point to values that do not exist any more.

\begin{example}\label{ex:update}Consider the input value
\([1, 2, 2, 3]\) and the filter ``\(.[] \models g\)'', where \(g\) is
``\(\ifj . = 2 \thenj \emptys \elsej .\)'', which we might suppose to
delete all values equal to 2 from the input list. However, the output of
jq is \([1, 2, 3]\). What happens here is perhaps unexpected, but
consistent with the above explanation of jq's semantics: jq builds the
paths \(.[0]\), \(.[1]\), \(.[2]\), and \(.[3]\). Next, it applies \(g\)
to all paths. Applying \(g\) to \(.[1]\) removes the first occurrence of
the number 2 from the list, leaving the list \([1, 2, 3]\) and the paths
\(.[2]\), \(.[3]\) to update. However, \(.[2]\) now refers to the number
3, and \(.[3]\) points beyond the list.\end{example}

Even if this particular example can be executed correctly with a special
case for filters that do not return exactly one output, there are more
general examples which this approach treats in unexpected ways.

\begin{example}\label{ex:update-comma}Consider the input value \([[0]]\)
and the filter ``\((.[],\; .[][]) \models g\)'', where \(g\) is
``\(\ifj . = [0] \thenj [1, 1] \elsej .+1\)''. Executing this filter in
jq first builds the path \(.[0]\) stemming from ``\(.[]\)'', then
\(.[0][0]\) stemming from ``\(.[][]\)''. Next, executing \(g\) on the
first path yields the intermediate result \([[1, 1]]\). Now, executing
\(g\) on the remaining path yields \([[2, 1]]\), instead of \([[2, 2]]\)
as we might have expected.\end{example}

The general problem here is that the execution of the filter \(g\)
changes the input value, yet only the paths constructed from the initial
input are considered. This leads to paths pointing to the wrong data,
paths pointing to non-existent data (both occurring in
\cref{ex:update}), and missing paths (\cref{ex:update-comma}).

I now show different semantics that avoid this problem, by interleaving
calls to \(f\) and \(g\). By doing so, these semantics can abandon the
idea of paths altogether.

The semantics use a helper function that takes an input array \(v\) and
replaces its \(i\)-th element by the output of \(\sigma\) applied to it:
\[(.[i] \models \sigma)|^c_v = \begin{cases}
[\langle v_0, \dots, v_{i-1} \rangle + \sigma|^c_{v_i} + \langle v_{i+1}, \dots, v_n \rangle] & \text{if } v = [v_0, \dots, v_n] \text{ and } 0 \leq i < n \\
\bot & \text{otherwise}
\end{cases}\]

\begin{table}

\caption{Update semantics. Here, \(\$x'\) is a fresh variable.
\label{tab:update-semantics}}

\begin{tabular}{ll}

\toprule

\(\mu\) & \(\mu \models \sigma\)\tabularnewline

\midrule

\(\emptys\) & \(.\)\tabularnewline
\(.\) & \(\sigma\)\tabularnewline
\(f \mid g\) & \(f \models (g \models \sigma)\)\tabularnewline
\(f, g\) & \((f \models \sigma) \mid (g \models \sigma)\)\tabularnewline
\(f \as \$x \mid g\) & \(\reduce f \as \$x'\; (.;\; g[\$x' / \$x] \models \sigma)\)\tabularnewline
\(\ifj f \thenj g \elsej h\) & \(\reduce f \as \$x'\; (.;\; \ifj \$x' \thenj g \models \sigma \elsej h \models \sigma)\)\tabularnewline
\(.[f]\) & \(\reduce f \as \$x'\; (.;\; .[\$x'] \models \sigma)\)\tabularnewline
\(.[]\) & \([.[] \mid \sigma]\)\tabularnewline
\(x(f_1; \dots; f_n)\) & \(g[f_1 / x_1, \dots, f_n / x_n] \models \sigma\)
if \(x(x_1; \dots; x_n) \coloneqq g\)\tabularnewline

\bottomrule

\end{tabular}

\end{table}

The update semantics are given in \autoref{tab:update-semantics}. The
case for \(f \as \$x \mid g\) is slightly tricky: Here, the intent is
that \(g\) has access to \(\$x\), but \(\sigma\) does not. This is to
ensure compatibility with jq's original semantics, which execute \(\mu\)
and \(\sigma\) independently, so \(\sigma\) should not be able to access
variables bound in \(\mu\). In order to ensure that, we replace \(\$x\)
by a fresh variable \(\$x'\) and substitute \(\$x\) by \(\$x'\) in
\(g\).

\begin{example}Consider the filter
\(0 \as \$x \mid (1 \as \$x \mid .[\$x]) \models \$x\). This updates the
input array at index \(1\). If the right-hand side of ``\(\models\)''
had access to variables bound on the right-hand side, then the array
element would be replaced by \(1\), because the variable binding
\(0 \as \$x\) would be shadowed by \(1 \as \$x\). However, because we
enforce that the right-hand side does not have access to variables bound
on the right-hand side, the array element is replaced by \(0\), which is
the value originally bound to \(\$x\). Given the input array
\([1, 2, 3]\), the filter yields the final result
\([1, 0, 3]\).\end{example}

In summary, the given semantics are easier to define and to reason
about, while keeping compatibility with the original semantics in most
use cases. Furthermore, avoiding to construct paths also appears to be
more performant, as I will show in \autoref{evaluation}.

\hypertarget{implementation}{%
\section{Implementation}\label{implementation}}

I implemented an interpreter for the specification in
\autoref{specification}. This interpreter is called \emph{jaq} and is
written in Rust \citep{DBLP:phd/dnb/Jung20}. Rust is a functional
systems programming language that focusses on memory safety and
performance.

jaq is divided into several parts, namely a parser, an interpreter, and
a command-line interface (CLI). The interpreter and the parser do not
use Rust's standard library and can be used independently from the CLI,
which allows their integration into other projects.

The memory allocator has a significant impact on the execution speed. In
particular, jaq uses the memory allocator mimalloc
\citep{DBLP:conf/aplas/LeijenZM19} as it has proven to improve
performance compared to the standard memory allocator.

In this section, we will first have a look at Rust's approach to sharing
data (\autoref{sharing}). Then, we discuss the data types that represent
filters (\autoref{filters}) and values (\autoref{values}). Finally, we
analyse the cloning behaviour of several filters (\autoref{cloning}),
and when sacrificing lazy for strict evaluation can contribute to higher
performance (\autoref{strictness}).

\hypertarget{sharing}{%
\subsection{Sharing}\label{sharing}}

Rust, like C/C++, does not implicitly share data and does not have a
garbage collector. That means that by default, duplicating a variable
creates a (deep) copy of its data in memory. In contrast, in many
garbage-collected languages, such as Haskell and OCaml, assigning
variables to other ones does not copy the data the variables refer to.
If we wish to share data in such a way in Rust, we can do so by wrapping
it with a reference-counting pointer type: For example,
\texttt{Rc\textless{}T\textgreater{}}, which is the counterpart of C++'s
\texttt{shared\_ptr\textless{}T\textgreater{}}, allows us to share data
of type \texttt{T} and deallocates the data when there is no more
reference to it.

\autoref{lst:rc} shows an example usage of \texttt{Rc}. In particular,
it shows how \texttt{Rc} allows for constant-time copying of values, at
the price of allowing mutating its data only when it has exclusive
access. The fact that Rust prevents us from using a value that has been
moved to an \texttt{Rc} (like \texttt{huge\_vec} in \autoref{lst:rc}) is
just one of the many safety nets of the language that ensure memory
safety. In languages such as C/C++, similar code might cause undefined
behaviour at runtime.

Later in this section, we will see that we have an interest to both
share data and potentially mutate it. To allow for mutation as much as
possible, the lesson here is the following: (1) We should aim to avoid
cloning \texttt{Rc}-shared values, in order to retain exclusive access
to the contained data. (2) If we must clone, we should at least strive
to limit the scope of the clone, in order to regain exclusive access to
the data (like in \autoref{lst:rc} after \texttt{b} has gone out of
scope). We will discuss how to fulfil (1) in \autoref{cloning} and (2)
in \autoref{strictness}.

\begin{listing}
\caption{Example usage of \texttt{Rc}.}

\hypertarget{lst:rc}{%
\label{lst:rc}}%
\begin{Shaded}
\begin{Highlighting}[]
\KeywordTok{let}\NormalTok{ huge\_vec }\OperatorTok{=} \DataTypeTok{Vec}\PreprocessorTok{::}\NormalTok{from([}\DecValTok{0}\OperatorTok{;} \DecValTok{100\_000}\NormalTok{])}\OperatorTok{;}
\CommentTok{// copying a \textasciigrave{}Vec\textasciigrave{} takes linear time}
\KeywordTok{let}\NormalTok{ copy\_vec }\OperatorTok{=}\NormalTok{ huge\_vec}\OperatorTok{.}\NormalTok{clone()}\OperatorTok{;}

\CommentTok{// move \textasciigrave{}huge\_vec\textasciigrave{} to a reference{-}counted pointer ...}
\KeywordTok{let} \KeywordTok{mut}\NormalTok{ a }\OperatorTok{=} \PreprocessorTok{Rc::}\NormalTok{new(huge\_vec)}\OperatorTok{;}
\CommentTok{// ... so trying to access it afterwards would throw a compiler error}
\CommentTok{//let len = huge\_vec.len();}

\CommentTok{// \textasciigrave{}a\textasciigrave{} has exclusive access to its data ...}
\PreprocessorTok{assert\_eq!}\NormalTok{(}\PreprocessorTok{Rc::}\NormalTok{strong\_count(}\OperatorTok{\&}\NormalTok{a)}\OperatorTok{,} \DecValTok{1}\NormalTok{)}\OperatorTok{;}
\CommentTok{// ... therefore we could mutate the data in \textasciigrave{}a\textasciigrave{}}
\PreprocessorTok{assert!}\NormalTok{(}\PreprocessorTok{Rc::}\NormalTok{get\_mut(}\OperatorTok{\&}\KeywordTok{mut}\NormalTok{ a)}\OperatorTok{.}\NormalTok{is\_some())}\OperatorTok{;}

\OperatorTok{\{} \CommentTok{// we open a new scope here}
    \CommentTok{// copying an \textasciigrave{}Rc\textasciigrave{} takes constant time, as its data is shared}
    \KeywordTok{let}\NormalTok{ b }\OperatorTok{=} \PreprocessorTok{Rc::}\NormalTok{clone(}\OperatorTok{\&}\NormalTok{a)}\OperatorTok{;}
    \CommentTok{// however, now \textasciigrave{}a\textasciigrave{} has lost exclusive access to its data ...}
    \PreprocessorTok{assert\_eq!}\NormalTok{(}\PreprocessorTok{Rc::}\NormalTok{strong\_count(}\OperatorTok{\&}\NormalTok{a)}\OperatorTok{,} \DecValTok{2}\NormalTok{)}\OperatorTok{;}
    \CommentTok{// ... therefore we cannot mutate the data in neither \textasciigrave{}a\textasciigrave{} nor \textasciigrave{}b\textasciigrave{}}
    \PreprocessorTok{assert!}\NormalTok{(}\PreprocessorTok{Rc::}\NormalTok{get\_mut(}\OperatorTok{\&}\KeywordTok{mut}\NormalTok{ a)}\OperatorTok{.}\NormalTok{is\_none())}\OperatorTok{;}
\OperatorTok{\}} \CommentTok{// \textasciigrave{}b\textasciigrave{} is going out of scope and is destroyed ...}
\CommentTok{// ... therefore we regain exclusive access to \textasciigrave{}a\textasciigrave{} and could mutate it again}
\PreprocessorTok{assert\_eq!}\NormalTok{(}\PreprocessorTok{Rc::}\NormalTok{strong\_count(}\OperatorTok{\&}\NormalTok{a)}\OperatorTok{,} \DecValTok{1}\NormalTok{)}\OperatorTok{;}
\end{Highlighting}
\end{Shaded}

\end{listing}

\hypertarget{filters}{%
\subsection{Filters}\label{filters}}

jaq represents filters as abstract syntax trees, in which variables are
mapped to de Bruijn indices \citep{deBruijn72}. This allows us to
implement the context as a list whose \(n\)-th element contains the
value bound to the variable \(n\). The context is an \texttt{Rc}-shared
linked list, which enables fast cloning and appending at the front. In
general, \(n\) is quite small, so lookup time is not an issue.

\emph{Core filters}, such as \texttt{length} or \texttt{sort}, are
implemented in the implementation language (C for jq and Rust for jaq).
In general, core filters execute significantly faster than filters
implemented by definition. Unlike jq, jaq does currently not support
defining \emph{recursive} filters. Therefore, several filters that are
implemented by recursive definition in jq, such as
\(\recurse(f) \coloneqq ., (f \mid \recurse(f))\), are implemented as
core filters in jaq. On the other hand, this allows jaq to inline all
calls to defined filters, resulting in one large filter without
definitions.

The implementation provides a function to evaluate filters: It takes a
filter \(\varphi\), a context \(c\), and a value \(v\), and yields an
\texttt{Iterator} over value results. Both the type and the
implementation of this function closely correspond to the evaluation
semantics \(\varphi|^c_v\) shown in \autoref{semantics}. Filters can be
safely shared across threads. This allows to evaluate any filter with
multiple inputs concurrently.

The \texttt{Iterator} trait from Rust's core library provides many
useful functions that enable a concise and performant implementation of
the semantics. For example, \texttt{Iterator} provides a lower and an
(optional) upper bound on the number of its elements (via the function
\texttt{size\_hint}). These bounds are inferred automatically for the
filter output \texttt{Iterator} from the implementation of the
evaluation function.

\begin{example}The evaluation of the concatenation \(f, g\) is
implemented using the existing \texttt{Iterator} combinator
\texttt{chain} in one line of code. The resulting \texttt{Iterator}
automatically derives its bounds; for example, the lower bound on the
number of elements yielded by \(f, g\) is the sum of the lower bounds on
the number of elements yielded by \(f\) and \(g\).\end{example}

Generally, the bounds are not tight because the number of elements
yielded by a filter cannot always be predicted (as filters are
Turing-complete!). However, for many simple filters, \texttt{Iterator}
yields tight bounds.

The bounds are used among others when collecting the output of a filter
\(f\) into an array by ``\([f]\)'', where the bounds on the number of
outputs yielded by \(f\) are used to pre-allocate memory for the output
array. In \autoref{cloning}, we will see another usage of
\texttt{Iterator} bounds.

\hypertarget{values}{%
\subsection{Values}\label{values}}

A shortened implementation of the value type is shown in
\autoref{lst:val}. It states that a value is either a boolean, an
integer, or an array.\footnote{Compared with this simplified value type,
  an actual JSON value can additionally be \texttt{null}, a decimal
  number, a string, or an associative map from strings to values.} In
the array case, \texttt{Rc\textless{}T\textgreater{}} is a
reference-counted pointer to data of type \texttt{T}. Using \texttt{Rc}
for arrays (and also for other types, such as strings and maps) is
crucial because otherwise, duplicating a \texttt{Val} would take time
linear in the size of the value, whereas with \texttt{Rc}, it takes
constant time.

\begin{listing}
\caption{Shortened version of the value type.}

\hypertarget{lst:val}{%
\label{lst:val}}%
\begin{Shaded}
\begin{Highlighting}[]
\KeywordTok{enum}\NormalTok{ Val }\OperatorTok{\{}
\NormalTok{    Bool(}\DataTypeTok{bool}\NormalTok{)}\OperatorTok{,}
\NormalTok{    Int(}\DataTypeTok{isize}\NormalTok{)}\OperatorTok{,}
\NormalTok{    Arr(Rc}\OperatorTok{\textless{}}\DataTypeTok{Vec}\OperatorTok{\textless{}}\NormalTok{Val}\OperatorTok{\textgreater{}\textgreater{}}\NormalTok{)}\OperatorTok{,}
\OperatorTok{\}}
\end{Highlighting}
\end{Shaded}

\end{listing}

As mentioned in \autoref{sharing}, there are situations in which we
would like to mutate values: Consider the filter
``\texttt{.\ +\ {[}0{]}}'', which checks that the input is an array and
returns its concatenation with \texttt{{[}0{]}}. It is tempting to
mutate the input array by concatenating it with \texttt{{[}0{]}}.
However, this might not be possible because there might be multiple
references to the input array. A safe option would be to create a new
array from the input array so that there is only a single reference to
it, then concatenate it with \texttt{{[}0{]}}. This is a correct, yet
very inefficient solution, because repeated concatenation with \(n\)
singleton arrays requires quadratic runtime due to creating \(n\) array
copies. We can solve this problem by cloning \emph{conditionally}: If
there is only a single reference to the array, then we mutate it,
otherwise we bite the bullet and clone the array before mutating it.
This pattern is provided by Rust's \texttt{Rc} type in form of the
function \texttt{make\_mut}. This entails that we have an interest to
avoid keeping references to values in order to allow for clone-free
mutation as much as possible. We will discuss the cloning behaviour of
several elementary filters in \autoref{cloning}.

\hypertarget{cloning}{%
\subsection{Cloning}\label{cloning}}

In this section, we will inspect the cloning behaviour of three
elementary filters, namely composition, concatenation, and binding. We
furthermore discuss a way to avoid clones in bindings, thus fulfilling
goal (1) laid out in \autoref{sharing}. The behaviour of the filters
described here coincides with the semantics in \autoref{semantics}.

We visualise filters as follows: A block represents a filter, taking one
value as input and yielding arbitrarily many values as output. A block
with rounded corners represents a filter that we are about to describe.
A line with an arrow represents a single value. A doubled line (such as
the line going into \(f\) in \autoref{fig:concat}) represents a cloned
value.

\begin{figure}
\begin{minipage}{.5\textwidth}
  \centering
  \includegraphics{tikz/compose.tex}
  \captionof{figure}{Composition of filters $f$ and $g$.}
  \label{fig:compose}
\end{minipage}%
\begin{minipage}{.5\textwidth}
  \centering
  \includegraphics{tikz/concat.tex}
  \captionof{figure}{Concatenation of filters $f$ and $g$.}
  \label{fig:concat}
\end{minipage}
\end{figure}

Let us start with the visualisation of composition in
\autoref{fig:compose}. Here, the input value to \(f \mid g\) is passed
to \(f\), and each of the outputs of \(f\) is passed to \(g\), yielding
the concatenation of the outputs of \(g\) as output of \(f \mid g\). No
cloning is required for this filter, as every input and output value is
used at most once.

Next, we inspect concatenation in \autoref{fig:concat}. Here, the input
of \(f, g\) is passed to both \(f\) and \(g\), and the output of
\(f, g\) is the concatenation of the outputs of \(f\) and \(g\). The
filter \(f\) receives a clone of the input value, whereas \(g\) receives
the original input value. This means that to mutate the input in \(f\),
we \emph{necessarily} have to clone it, whereas we might be able to
clone-freely mutate the input in \(g\) (depending on the existence of
other references to the input).

Note that this restriction is a consequence of our assumption that any
filter may access its input.

\begin{example}Consider the filter ``\((. + [0]), 1\)''. Here, we pass
the cloned input to ``\(. + [0]\)'', requiring a clone of the input
array to concatenate \([0]\) to it. Furthermore, we pass the original
input to ``\(1\)''. However, the filter ``\(1\)'' does clearly not
depend on its input, so in principle, we could pass the original input
to ``\(. + [0]\)'', allowing it to clone-freely mutate the
input.\end{example}

However, distinguishing filters using their input from filters not using
it would significantly complicate the implementation, which is why I
refrained from doing so.

\begin{figure}
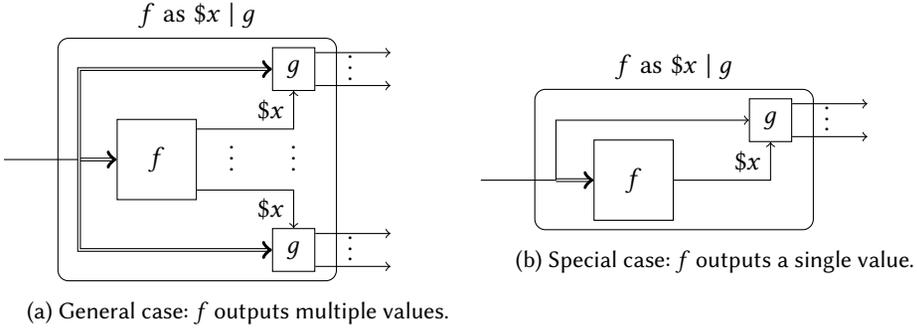

\begin{subfigure}{0.45\textwidth}
  \includegraphics{tikz/bind-many.tex}
  \caption{General case: $f$ outputs multiple values.}
  \label{fig:bind-many}
\end{subfigure}
\begin{subfigure}{0.45\textwidth}
  \includegraphics{tikz/bind-single.tex}
  \caption{Special case: $f$ outputs a single value.}
  \label{fig:bind-single}
\end{subfigure}
\caption{Binding the output of a filter $f$ as $\$x$ in $g$.}
\label{fig:bind}
\end{figure}

Let us continue with the last filter, namely binding, shown in
\autoref{fig:bind}. The general case for \(f \as \$x \mid g\), where we
do not know in advance how many outputs \(f\) yields, is shown in
\autoref{fig:bind-many}. Here, the input is passed as clone to both
\(f\) and all instances of \(g\), and every output of \(f\) is bound to
\(\$x\) and passed to \(g\). That means that we cannot clone-freely
mutate the input neither in \(f\) nor in \(g\). What if we could bound
the number of outputs of \(f\) in advance to \(n\)? In that case, we
could reduce the number of clones from \(n+1\) in the general case to
\(n\). This is because once the last value from \(f\) has been yielded,
we could pass the original input to the last instance of \(g\). The
implementation implements this idea for the special case where we know
that \(f\) yields at most \(n = 1\) outputs, which is shown in
\autoref{fig:bind-single}. In this case, we pass the cloned input to
\(f\) and the original input to \(g\).

The question remains how to predict the number of outputs yielded by
\(f\). Fortunately, this is surprisingly easy: Executing a filter yields
an \texttt{Iterator} over its output values, and as discussed in
\autoref{filters}, \texttt{Iterator} provides lower and upper bounds on
the number of its elements. The code to execute \(f \as \$x \mid g\)
thus uses a special case when executing \(f\) yields at least and at
most one element.

\hypertarget{strictness}{%
\subsection{Strictness}\label{strictness}}

In this subsection, we will discuss strict evaluation of certain filters
in jaq, namely Cartesian operations such as addition. This makes jaq's
evaluation differ from the formal semantics (\autoref{semantics}), which
do not provide for strict evaluation. In return, it reduces the scope of
cloned values as formulated as goal (2) in \autoref{sharing}, thus
increasing the chances of clone-free mutation in order to yield higher
performance.

The Cartesian filter \(l \circ r\) yields the following values, where
\(r\) evaluates to \(r_1, \dots, r_n\) and \(l\) evaluates to
\(l_1, \dots, l_m\): \[\begin{matrix}
l_1 \circ r_1, & \dots, & l_1 \circ r_n, \\
\vdots & \ddots & \vdots \\
l_m \circ r_1, & \dots, & l_m \circ r_n
\end{matrix}\] Supposing that ``\(\circ\)'' is defined on values, we
could naively evaluate ``\(l \circ r\)'' by
``\(l \as \$x \mid r \as \$y \mid \$x \circ \$y\)'' (given that \(\$x\)
does not occur in \(r\)). This approach has two downsides: First, for
every value \(\$x\) yielded by \(l\), we have to re-evaluate \(r\), even
though it does not depend on \(\$x\). Second, as we can see in
\autoref{fig:bind-single}, \(l\) always operates on a cloned version of
the input. However, Cartesian operations \(l \circ r\) are biased to
only mutate \(l\), because arrays and strings are more efficiently
extended to the right. This means that this approach is detrimental to
performance because it prevents clone-free mutation of values yielded by
\(l\). Maybe for this reason, jq implements an approach similar to
``\(r \as \$y \mid l \as \$x \mid \$x \circ \$y\)'', which allows for
clone-free mutation of values yielded by \(l\), but only if \(r\) yields
a single output. Furthermore, this approach yields the output in a
different (and perhaps surprising) order. jaq implements a third
approach: It first collects \emph{all} outputs of \(r\) into an array
\(\$ys\). Only then, it evaluates \(l\) with the original input, and for
every output \(\$x\) of \(l\), for every value \(\$y\) in \(\$ys\), it
yields \(\$x \circ \$y\). Unlike the previous two approaches, this
approach does not duplicate work and enables clone-free mutation of
values yielded by \(l\) regardless of the number of outputs of \(r\).
Furthermore, unlike the approach taken in jq, this approach yields the
outputs in the same order as the first approach. However, the price it
pays is the strict evaluation of \(r\), which never terminates if \(r\)
yields an infinite sequence of outputs. Fortunately, when this is an
issue, we can always use the first or second approach in jaq.

\hypertarget{evaluation}{%
\section{Evaluation}\label{evaluation}}

I evaluated the performance of jq\footnote{jq was obtained from
  \url{https://github.com/stedolan/jq/}, rev.
  \href{https://github.com/stedolan/jq/tree/cff5336ec71b6fee396a95bb0e4bea365e0cd1e8}{cff5336}.
  This version significantly improves performance over the latest stable
  release, jq 1.6. In the \texttt{Makefile}, \texttt{-DNDEBUG} was added
  to \texttt{DEFS} in order to omit assertion checks.} and jaq 0.8.1 by
measuring their execution runtime on a hand-crafted set of filters.
These filters cover a large part of the filters introduced in
\autoref{preliminaries}.

I performed the evaluation on Ubuntu 22.04, running on a machine with
four Intel Core i3-5010U CPUs à 2.10 GHz and 8 GB RAM. jq and jaq were
compiled with GCC 11.2.0 and Rust 1.62.0. I exploited only a single core
for the evaluation.

All benchmarks are parametrised by an input integer \(n\). The benchmark
\texttt{b} parametrised with \texttt{n} is referred to as
\texttt{b}-\(n\); for example, \texttt{empty}-128.

There are two special benchmarks:

\begin{itemize}
\tightlist
\item
  \texttt{empty}: This benchmark consists of the filter \texttt{empty}
  that is called \(n\) times with no input. It serves to measure the
  overhead of starting the program.
\item
  \texttt{bf-fib}: This benchmark consists of a Brainfuck interpreter
  written in jq that is ran on a Brainfuck program evaluating \(n\)
  Fibonacci numbers.
\end{itemize}

The other benchmarks all call a filter a single time, passing \(n\) as
input value. Because their produced values are rather large, the output
of the filter is passed through the filter \texttt{length}, in order not
to measure I/O speed.

\begin{itemize}
\tightlist
\item
  \texttt{reverse}: The filter
  ``\texttt{{[}range(.){]}\ \textbar{}\ reverse}'' reverses an array
  \([0, \dots, n-1]\). Here, ``\texttt{range(n)}'' returns the sequence
  \(0, 1, \dots, n-1\).
\item
  \texttt{sort}: The filter
  ``\texttt{{[}range(.)\ \textbar{}\ -.{]}\ \textbar{}\ sort}'' sorts an
  array of decreasing numbers \([0, -1, \dots, -(n-1)]\).
\item
  \texttt{add}: The filter
  ``\texttt{{[}range(.)\ \textbar{}\ {[}.{]}{]}\ \textbar{}\ add}''
  constructs an array of arrays where each contains a single number,
  then concatenates all singleton arrays into one large array with
  \texttt{add}.
\item
  \texttt{kv}: The filter \texttt{kv}, which is defined as
  ``\texttt{{[}range(.)\ \textbar{}\ \{(tostring):\ .\}{]}\ \textbar{}\ add}'',
  constructs an object \texttt{\{"0":\ 0,\ "1":\ 1,\ ...\}} consisting
  of \(n\) key-value pairs.
\item
  \texttt{kv-update}: The filter
  ``\texttt{kv\ \textbar{}\ .{[}{]}\ +=\ 1}'' increments the values of
  the object constructed with \texttt{kv}. Here, ``\texttt{p\ +=\ f}''
  is short-hand for ``\texttt{p\ \textbar{}=\ .\ +\ f}''.
\item
  \texttt{kv-entries}: The filter
  ``\texttt{kv\ \textbar{}\ with\_entries(.value\ +=\ 1)}'' performs the
  same as \texttt{kv-update}, but by deconstructing the object into an
  array of key-value pairs, mapping over them, then reconstructing an
  object.
\item
  \texttt{ex-implode}: The filter
  ``\texttt{{[}limit(.;\ repeat("a")){]}\ \textbar{}\ add\ \textbar{}\ explode\ \textbar{}\ implode}''
  first constructs a string consisting of \(n\) occurrences of ``a'',
  then deconstructs the string into an array of its characters with
  \texttt{explode}, then reconstructs the array into a string with
  \texttt{implode}.\footnote{The filter
    \texttt{explode\ \textbar{}\ implode} is the identity on strings.}
\item
  \texttt{reduce}: The filter
  ``\texttt{reduce\ range(.)\ as\ \$x\ ({[}{]};\ .\ +\ {[}\$x\ +\ .{[}-1{]}{]})}''
  yields an array \([x_0, \dots, x_n]\), where \(x_0 = 0\) and
  \(x_i = i + x_{i - 1}\).
\item
  \texttt{tree-flatten}: The filter
  ``\texttt{nth(.;\ 0\ \textbar{}\ trees)\ \textbar{}\ flatten}''
  constructs a binary tree of depth \(n\) having 0 as leafs, then
  obtains all of its \(2^n\) leafs. Here, the filter \texttt{trees}
  expands to ``\texttt{recurse({[}.,\ .{]})}''.
\item
  \texttt{tree-update}: The filter
  ``\texttt{nth(.;\ 0\ \textbar{}\ trees)\ \textbar{}\ (..\ \textbar{}\ scalars)\ \textbar{}=\ .+1}''
  constructs a binary tree as in the previous example, then updates all
  leafs by incrementing them.
\item
  \texttt{to-fromjson}: The filter
  ``\texttt{"{[}"\ +\ ({[}range(.)\ \textbar{}\ tojson{]}\ \textbar{}\ join(","))\ +\ "{]}"\ \textbar{}\ fromjson}''
  first constructs a string that encodes the array \([0, \dots, n-1]\),
  then parses it into the array.
\end{itemize}

\begin{figure}
\includegraphics{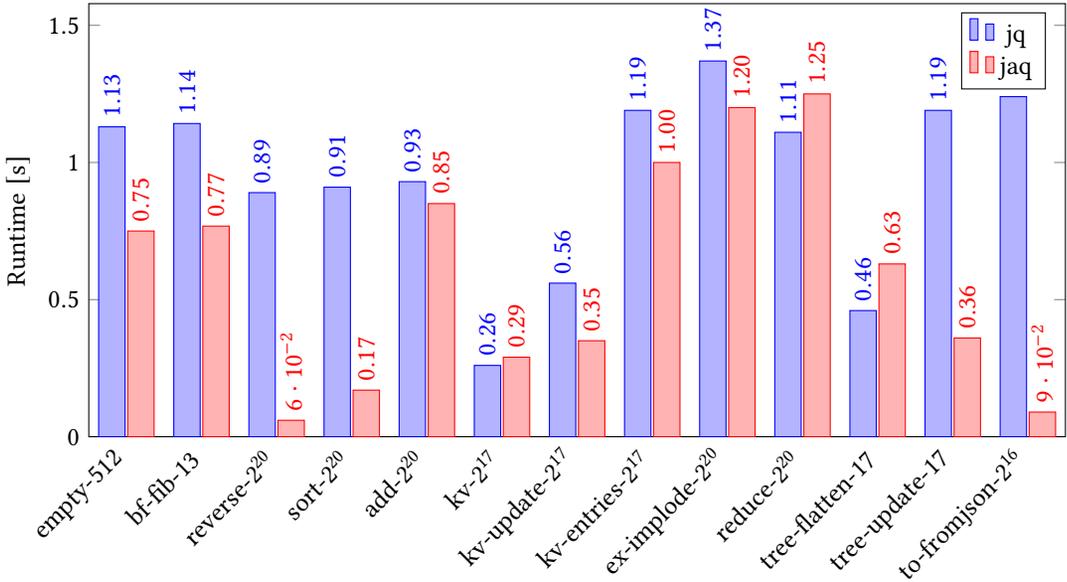}
\caption{Runtime of jq and jaq for a set of benchmarks (lower is
better).}
\label{fig:runtime}
\end{figure}

The results are shown in \autoref{fig:runtime}. In ten out of thirteen
benchmarks, jaq is faster than jq. Some of these benchmark results where
jaq is faster can be attributed to the efficient data structures in
Rust's core library, such as \texttt{Vec}, which shines in the
\texttt{reverse} and \texttt{sort} benchmarks.

One benchmark that shows the power of the new update semantics is
\texttt{to-fromjson}: In this benchmark, the operation that takes most
time in jq is \texttt{join}. In jq, a simplified version\footnote{jq
  actually uses a more complex version of \texttt{join} that converts
  the list elements to strings depending on their type, but this is
  neither essential to the example nor significantly impacts
  performance.} of \texttt{join(sep)} is defined by
``\texttt{reduce\ .{[}{]}\ as\ \$x\ (null;\ if\ .\ ==\ null\ then\ \$x\ else\ .\ +\ sep\ +\ \$x\ end)}'',
whereas in jaq, it is defined by
``\texttt{(.{[}:-1{]}{[}{]}\ +=\ sep)\ \textbar{}\ add}''. The jq
version folds over the list, whereas the jaq version retrieves all list
elements except for the last (``\texttt{.{[}:-1{]}{[}{]}}''), updates
them by adding \texttt{sep} to them, then concatenates all list elements
with \texttt{add}. We can also use the jaq version of \texttt{join} in
jq; however, this is greatly slower than jq's version. Most likely, jq
builds an exhaustive list of indices before performing the update, which
we can omit in jaq due to its update semantics (\autoref{updates}).

I additionally benchmarked a jq implementation written in Go, namely
gojq 0.12.9\footnote{Retrieved from
  \url{https://github.com/itchyny/gojq/}.}. gojq is faster than both jaq
and jq only on one benchmark, namely \texttt{tree-flatten}, because
gojq's implementation of the filter \texttt{flatten} is written in Go,
not in jq. gojq is also faster than jq on \texttt{empty} and
\texttt{to-fromjson}-\(2^{16}\). On the other benchmarks, gojq is slower
than both jaq and jq.

Yet another implementation, namely yq, was not evaluated because its
syntax is not compatible with that of jq.

\hypertarget{conclusion}{%
\section{Conclusion}\label{conclusion}}

I showed formal syntax and denotational semantics for a subset of the jq
language that significantly eases reasoning about the behaviour of jq
filters. In particular, I showed shortcomings in the way the existing jq
implementation handles updates with paths, and introduced new update
semantics that avoid these pitfalls. I implemented the existing syntax
and the new semantics in an interpreter called jaq, which is heavily
built around Rust and its \texttt{Iterator} trait. In order to allow for
clone-free mutation of values, care is taken in jaq to avoid cloning
where possible, such as in the evaluation of variable bindings. The
evaluation shows that jaq is faster than jq on ten out of thirteen
benchmarks, despite having a significantly simpler execution model. This
indicates that we can have the best of many worlds, namely formal
semantics, a small implementation, and state-of-the-art performance.

\begin{acks}
Thanks to Diana Gründlinger and Fabian Mitterwallner for their valuable
comments on drafts of this paper. This research was funded in part by
the Austrian Science Fund (FWF) {[}J~4386{]}.
\end{acks}

\bibliography{literature.bib}

\end{document}